\documentclass[oneside]{article}

\usepackage{multicol}
\usepackage{macros}
\usepackage{graphicx,hyperref}
\usepackage{amsmath,amsthm}
\usepackage[utf8]{inputenc}
\usepackage{microtype}
\usepackage[charter,cal=cmcal]{mathdesign}
\usepackage{mathtools}
\usepackage[portrait,letterpaper,margin=1in]{geometry}\usepackage[backend=biber,autocite=footnote,style=numeric]{biblatex}
\usepackage{enumitem}
\usepackage{xspace}
\usepackage{algorithm}
\usepackage{algpseudocode}

\DeclareCiteCommand{\footcite}[\mkbibfootnote]
  {\usebibmacro{prenote}}
  {\usebibmacro{citeindex}%
   \printtext[brackets]{\usebibmacro{cite}}}
  {\multicitedelim}
  {\usebibmacro{postnote}}

\theoremstyle{theorem}
\newtheorem{theorem}{Theorem}
\newtheorem{lemma}{Lemma}
\theoremstyle{definition}
\newtheorem{definition}{Definition}
\newtheorem{observation}{Observation}

\newtheorem{corollary}{Corollary}

\newcommand\restr[2]{{
  \left.\kern-\nulldelimiterspace 
  #1 
  \vphantom{|} 
  \right|_{#2} 
}}

\newcommand{\psm}[2]{\restr{#1}{#2}}

\DeclareMathOperator{\psupp}{supp_{+}}
\DeclareMathOperator{\supp}{supp}
\DeclareMathOperator{\reduce}{r}

\def\R{\mathbb{R}}

\def\One{\mathbb{1}}

\def\F{\mathbb{F}}
\def\wh{\widehat}

\DeclareMathOperator{\Char}{char}
\DeclareMathOperator{\outmap}{\mathsf{out}}
\DeclareMathOperator{\outdir}{\mathsf{outdir}}
\DeclareMathOperator{\out}{\mathsf{out}}
\def\xor{\oplus}
\def\varindex{\bullet}

\DeclareMathOperator{\LCP}{LCP}

\def\coNP{\textsf{co-NP}\xspace}

\newcommand{\amat}[2]{{\mathcal{C}(#1,#2)}}
\newcommand{\CC}[2]{\amat{#1}{#2}}

\title{On the computational equivalence of co-NP refutations of a matrix being a P-matrix}
\author{Spencer Gordon \and Kevin Shu}
\bibliography{bibliography.bib}

\begin{document}
\maketitle
\begin{abstract}
  A P-matrix is a square matrix $X$ such that all principal submatrices of $X$ have positive determinant.
  Such matrices appear naturally in instances of the linear complementarity problem, where these are precisely the matrices for which the corresponding linear complementarity problem has a unique solution for any input vector.
  Testing whether or not a square matrix is a P-matrix is \coNP complete, so while it is possible to exhibit polynomially-sized witnesses for the fact that a matrix is not a P-matrix, it is believed that there is no efficient way to prove that a given matrix is a P-matrix.
  We will show that several well known witnesses for the fact that a matrix is not a P-matrix are computationally equivalent, so that we are able to convert between them in polynomial time, answering a question raised in \cite{ueopl}.
\end{abstract}

A P-matrix is a matrix with the property that all of its principal minors are strictly positive.
If $X$ is symmetric, then $X$ is a P-matrix if and only if $X$ is positive semidefinite, meaning that all of its eigenvalues are positive. An immediate consequence is that we can check if a symmetric matrix is a P-matrix in polynomial time.
On the other hand, for nonsymmetric matrices, being a P-matrix is a much more difficult property to understand, and indeed, it is \coNP complete to check that general (nonsymmetric) matrix is a P-matrix \cite{coxson1994p}.

The P-matrix property has connections to a number of different problems in computer science and linear algebra; \cite{samelson1958partition} gives one of the earliest definitions of the P-matrix property for its connection to a property of arrangements of vectors in Euclidean space, where certain convex cones spanned by these vectors divide space up into disjoint pieces.

Stickney and Watson connected this property to the linear-complementarity problem in \cite{StickneyAlan1978DMoB}, which is a quadratic optimization problem of the form
\begin{equation} \label{lcp(M,q)} \tag{$\LCP(M,q)$}
  \begin{split}
    \text{Find } w,z \in \R^n &       \\
    \text{subject to } w - Mz & = q   \\
    w, z                      & \ge 0 \\
    w \cdot z                 & = 0.
  \end{split}
\end{equation}
The problem data is $(M,z) \in \R^{n\times n} \times \R^n$ and $w, z\in \R^n$ are the optimization variables.
This problem has connections to a number of problems in applied mathematics. 
This linear complementarity problem has a unique solution for any choice of $q$ if and only if $M$ is a P-matrix, though it is not yet known if there is a polynomial time algorithm for finding this unique solution even if $M$ is a P-matrix.

While it is not easy to check that a matrix is a P-matrix, we can still give some equivalent conditions for being a P-matrix.
The three equivalent definitions we will be interested in here will be
\begin{enumerate}
  \item Every principal submatrix of $M$ has positive determinant.
  \item For every vector $x\in \R^n$ so that $x\neq 0$, $x_i(Mx)_i > 0$ for all $i\in [n]$.
  \item The matrix $M$ defines a unique sink orientation on the hypercube for every vector $q\in \R^n$. \label{item:uso}
\end{enumerate}
We will define item \ref{item:uso} precisely in section \ref{sec:usos}.
A unique sink orientation (USO) is a type of orientation on a hypercube graph which models the combinatorial properties of linear optimization problems on the hypercube polytope.
USO's are of interest because it has been suggested that they may provide a path towards constructing strongly polynomial time algorithms for both the linear programming problem and the P-matrix linear complementarity problem.

The key thing to notice here is that these three definitions are \coNP type definitions, i.e. they each allow us to produce polynomially sized proofs that an input matrix is not a P-matrix, which can be checked in polynomial time.
One difficulty in approaching this subject is that the proofs of these equivalences are spread out in various places in the literature.
Moreover, they are not explicitly constructive, meaning that it is not clear that given a witness that a matrix violates one of these conditions, we can easily provide a witness showing that it violates another of these conditions.

Our goal is to collect the proofs of these equivalences into one document, and also make these equivalences explicitly constructive, in that we will provide polynomial time algorithms which take in a witness that a matrix violates one of these conditions and outputs a witness that that matrix violates the other conditions.

We will fix our notation and define these conditions more precisely in \ref{sec:prelim}, and then give our proofs of these computational equivalences in sections \ref{sec:minor2vector} and \ref{sec:minor2uso}.
We hope that our exposition will be easy to understand as an introduction to the subject of P-matrices, as well as to USO's.

\section{Preliminaries}\label{sec:prelim}
\subsection{Definitions and notation}
For any fixed $n\in \N$, let $e_1,\dotsc,e_n$ denote the standard basis vectors for $\R^n$, and let $[n] \coloneqq \Set{1,2,\dotsc,n}$. For any subset $\alpha \subseteq [n]$, let $\bar{\alpha} \coloneqq [n] \setminus \alpha$.
\begin{definition}[Matrix Indexing] For any matrix $M\in \R^{m\times n}$ we'll use $M_{\varindex j}$ to denote the $j$th column of $M$ and $M_{i \varindex}$ to denote the $i$th row of $M$.	 When $\alpha,\beta \subseteq [n]$, we'll extend the notation so that $M_{\alpha \beta}$ is the submatrix with rows in $\alpha$ and columns in $\beta$. To avoid clutter, we'll define $\restr{M}{\alpha} \coloneqq M_{\alpha\alpha}$ as shorthand for principal submatrices.
\end{definition}

\begin{definition}[The complementary cones induced by $M$]  For any matrix $M \in \R^{n\times n}$ and $\alpha \subseteq [n]$, let $\CC{\alpha}{M}$ be the matrix with \[ \CC{\alpha}{M}_{\varindex j} = \begin{cases} -M_{\varindex j} & \quad\text{if $j \in \alpha$}   \\
              e_j              & \quad\text{if $j\notin \alpha$}\end{cases} \]
\end{definition}

\begin{definition}[]The Identity Matrix] Let $I_{k}$ denote the $k\times k$ identity matrix.
\end{definition}

\begin{definition}[Identifying subsets with boolean vectors] To each subset $\alpha \subseteq [n]$, let $\Char(\alpha) \in \F_2^n$ be the characteristic vector of $\alpha$. Given subsets	$\alpha, \beta \subseteq [n]$, we'll write $\alpha \oplus \beta \coloneqq \Set{i \in (\alpha\setminus \beta) \cup (\beta \setminus \alpha)}$ by analogy with addition in $\F_2^n$. For $\beta \subseteq [n]$ and $i\in [n]$, we'll write $\beta \pm i$ instead of $\beta \pm \Set{i}$. \end{definition}

\begin{definition}[Sign function] We'll use the sign function $\sgn: \R \to \Set{-1,0,+1}$ given by $\sgn(x) = 1$ for $x > 0$, $\sgn(x) = -1$ for $x \leq 0$.
\end{definition}

\begin{definition}[Support of a vector] Given a vector $x \in \R_{+}$, let $\psupp(x) = \Setbar{i}{x_i > 0}$.
\end{definition}

\begin{definition}[Non-degenerate matrix] A matrix $M \in \R^{n\times n}$ is said to be \emph{non-degenerate} if all of its principal minors are non-zero, i.e. $\det(M_{\alpha\alpha}) \neq 0$ for all $\alpha \subseteq [n]$.
\end{definition}
Non-degeneracy ensures that for any choice of $\alpha \subseteq [n]$, $\amat{\alpha}{M}$ is invertible:
\begin{lemma}
    The following are equivalent:
    \begin{enumerate}
        \item $M$ is non-degenerate.
        \item For all $\alpha \subseteq [n]$, $\det (\amat{\alpha}{M}) \neq 0$.
    \end{enumerate} \label{lem:nondegenerate}
\end{lemma}

\begin{proof}
    Let $\alpha \subseteq [n]$, and by permuting the rows, we may as well take  $\alpha$ as $\{1 ,\dots, k\}$ and consider $\amat{\alpha}{M}$ in block form:
    \[
        \begin{pmatrix}
            -M|_{\alpha}            & 0     \\
            -M_{\bar{\alpha}\alpha} & I_{k}
        \end{pmatrix}
    \]
    By the formula for the determinant of a block lower triangular matrix, the determinant of
    \[
        \amat{\alpha}{M} = (-1)^k\det(M_{\alpha})\det(I_k).
    \]
    In particular, $\amat{\alpha}{M}$ is singular if and only if $M_{\alpha}$ is.

\end{proof}

\begin{definition}[The outmap and out directions]
    Now for any matrix $M \in \R^{n\times n}$, subset $\alpha \subseteq [n]$ for which $\amat{\alpha}{M}$ is invertible, and vector $q\in \R^n$, we'll define $\outmap_q(\alpha,M) \in \Set{0,1}^n$ as follows:
    \[ \outmap_q(\alpha, M)_i = \begin{cases}
            1 & \quad\text{if $(\CC{\alpha}{M}^{-1}q)_i > 0$}    \\
            0 & \quad\text{if $(\CC{\alpha}{M}^{-1}q)_i \leq 0$}
        \end{cases} \]

    Finally, we'll define $\outdir_q(\alpha, M) \coloneqq \psupp(\outmap_q(M, \alpha))$. By Lemma~\ref{lem:nondegenerate}, the outmap is well-defined for all $\alpha$ if and only if $M$ is non-degenerate.
\end{definition}



\subsection{The three different witnesses}
A matrix $M \in \R^{n\times n}$ is not a $P$-matrix if and only if we can find certificates of the following form:
\begin{enumerate}[label=(PV\arabic*)]
    \item \emph{Non-positive minor.} A subset $\alpha \subseteq [n]$ such that $\det(\restr{M}{\alpha}) \leq 0$. \label{PV1}
    \item \emph{Sign-reversing vector.} A vector $x\in \R^n$, $x\neq 0$, such that $x_i(Mx)_i \leq 0$ for all $i\in [n]$. \label{PV2}
    \item \emph{A failed unique sink orientation.} Either a set $\alpha \subseteq [n]$ such that $\amat{\alpha}{M}$ is not invertible or a vector $q\in \R^n$ and two distinct sets $\alpha, \beta \subseteq [n]$ such that \[\Brack{\Char(\alpha) \xor \Char(\beta)} \cap \Brack{\outmap_q(\alpha, M) \xor \outmap_q(\beta, M)} = 0^n\text{.} \] \label{PV3}
\end{enumerate}

\subsection{Unique Sink Orientations}\label{sec:usos}

The hypercube graph $H^n$ is a graph with vertex set $2^{[n]}$ where there is an edge between $\alpha$ and $\beta$ exactly when $\Card{\alpha \oplus \beta} = 1$.


Let $S, T \subseteq [n]$. To each such pair of subsets we associate a subgraph of $H^n$, which is induced by the vertex set
\[
    [S,T] = \Set{\alpha : S \subseteq \alpha \subseteq T}.
\]
We will denote this induced subgraph by $H^n[S,T]$, and we will refer to the subgraph $H^n[S,T]$ as a face of $H^n$.

An orientation of the graph $H^n$ is an assignment of a direction to each edge of $H^n$, i.e. a map $\mathcal{O} : E(H^n) \rightarrow V(H^n) \otimes V(H^n)$, where $\mathcal{O}(\{\alpha, \beta\}) = (\alpha, \beta)$ is interpreted as meaning that the edge $\{\alpha, \beta\}$ is directed from $\alpha$ to $\beta$.

A unique sink orientation (USO) on the hypercube graph is an assignment of a direction to each edge of $H^n$ so that for each $S, T \subseteq [n]$, there is a unique vertex $v$ in $H^n[S,T]$ so that all edges in $H^n[S,T]$ are directed towards $v$.
This vertex $v$ is called the sink of the face $H^n[S,T]$.


The unique sink orientation problem is to find the unique sink of a given unique sink orientation.

One natural way to obtain a USO on $H^n$ is by considering the polytope $Q^n$, which is the convex hull of the vectors $\Set{\Char(\alpha):\alpha \subseteq[n]}$.
Two vertices $\Char(\alpha)$ and $\Char(\beta)$ are adjacent in $Q^n$ if and only if $\alpha$ and $\beta$ are adjacent in the graph $H^n$.

A linear function $f : \R^n \rightarrow \R$ assigns each vertex of $Q^n$ a value in $\R$. If all of the values assigned to the vertices of $Q^n$ are distinct (in which case, we will call $f$ nondegenerate), then we can obtain a USO on $H^n$ by directing the edge $\{\alpha, \beta\} \in E(H^n)$ to be $(\alpha, \beta)$ if and only if $f(x_{\beta}) > f(x_{\alpha})$.
This clearly produces a USO, and moreover, this direction of the edges is acyclic in the sense that there are no directed cycles in the graph.
Finding the global maximum value of $f$ on $Q^n$ can thus be viewed as an instance of the USO problem.
The simplex algorithm can be viewed combinatorially as a way of traversing this USO. General approaches to the USO problem, where the USO need not be acyclic, inspects vertices in a random access way, rather than viewing them sequentially (as in \cite{959931}).

The outmap associated with an orientation $\mathcal{O}$ of $H^n$ is a map $\out:2^{[n]} \rightarrow 2^{[n]}$, where $\out(\alpha) = \{i : \mathcal{O}(\{\alpha, \alpha \oplus \{i\}\}) = (\alpha, \alpha \oplus \{i\}) \}$.
In \cite{959931}, it was shown that if $\mathcal{O}$ is a USO, the outmap is in fact a bijection, and it moreover satisfies the property that
\[
    (\alpha \oplus \beta) \cap (\mathcal{O}(\alpha) \oplus \mathcal{O}(\beta)) \neq \varnothing.
\]

Indeed, in \cite{959931}, it was shown that if $\out : 2^{[n]} \rightarrow 2^{[n]}$ satisfies $(\alpha \oplus \beta) \cap (\mathcal{O}(\alpha) \oplus \mathcal{O}(\beta)) \neq \varnothing$, then in fact, $\out$ is the outmap of some unique sink orientation.

This motivates our definition that a witness to the fact that $\out$ is not a USO is a pair $\alpha, \beta \subseteq 2^{[n]}$ so that $\alpha \neq \beta$, and
\[
    (\alpha \oplus \beta) \cap (\mathcal{O}(\alpha) \oplus \mathcal{O}(\beta)) = \varnothing.
\]

\subsection{USOs and P-matrix Linear Complementarity Problems}
A P-matrix also produces a USO through its associated linear complementarity problem.
In the linear complementarity problem, we can consider
\begin{align*}
    \text{Find } w,z \in \R^n &       \\
    \text{subject to } w - Mz & = q   \\
    w, z                      & \ge 0 \\
    w \cdot z                 & = 0.
\end{align*}
The constraints $w\cdot z = 0$ and $w, z \ge 0$ imply that for each $i \in [n]$, either $w_i = 0$ or $z_i = 0$.
Hence, there is a natural correspondence between the subsets $\alpha \subseteq [n]$ and potential solutions, namely, we can assume that $\supp(w) \subseteq \alpha$, and that $\supp(z) \subseteq \bar{\alpha}$.
If we assume that this support condition holds, then the only possible solutions to this problem are those that solve the linear system
\[
    \sum_{i \in \alpha} w_i e_i -
    \sum_{i \in \bar{\alpha}} z_i M_{\varindex i}  =
    \amat{\alpha}{M}
    \Paren{\sum_{i \in \alpha} w_i e_i +
        \sum_{i \in \bar{\alpha}} z_i e_i}  = q
\]
Hence, we see that this linear complementarity problem has a solution with these support constraints if and only if $\amat{\alpha}{M}^{-1}q$ is a nonnegative vector.

We will say that $q$ is nondegenerate if for every $\alpha \subseteq [n]$, and each $i \in [n]$, $(\amat{\alpha}{M}^{-1}q)_i \neq 0$.

If $q$ is nondegenerate, one can show that if $M$ is a P-matrix and $\alpha \oplus \beta = \{i\}$, then exactly one of $(\amat{\alpha}{M}^{-1}q)_i$ and $(\amat{\beta}{M}^{-1}q)_i$ is positive.
Hence, we can orient the edge $\{\alpha, \beta\} \in E(H^n)$ as $(\alpha, \beta)$ if $(\amat{\beta}{M}^{-1}q)_i > 0$, and it was shown in \cite{StickneyAlan1978DMoB} that this defines a unique sink orientation.

\section{Non-positive minors and sign-reversing vectors}\label{sec:minor2vector}
\subsection{Non-positive minors to sign-reversing vectors}

\begin{theorem}
    We can convert a violation of type \ref{PV1} to a violation of type \ref{PV2} in polynomial time.
\end{theorem}
\begin{proof} Given a subset $\alpha$ with $\det(\psm{A}{\alpha}) \leq 0$ as in \ref{PV1}, there are two cases.

    If $\det(\psm{A}{\alpha}) = 0$, we can simply use Gaussian elimination to find a non-zero vector $\wh{x} \in \ker(\psm{A}{\alpha}) \subseteq \R^{\alpha}$. Now let $x\in \R^n$ be the extension of $\wh{x}$ to $\R^n$ where the coordinates not in $\alpha$ are all zero. Then for $i\in \alpha$, $(Ax)_i = 0$ so $x_i(Ax)_i \leq 0$ and for $i\notin \alpha$, $x_i = 0$ so $x_i(Ax)_i \leq 0$ and $x$ is a violation of type \ref{PV2}.

    If $\det(\psm{A}{\alpha}) < 0$, then $\psm{A}{\alpha}$ must have a negative real eigenvalue, since the product of the eigenvalues of $\psm{A}{\alpha}$ is a negative real number and complex eigenvalues come in conjugate pairs. Consider an eigenvector $\wh{x} \in \R^{\alpha}$ for some negative eigenvalue.\footnote{Using inverse iteration, we can compute an approximation that is sufficient for our purposes in time $\poly(n,B)$ where $B$ is the sum of the bit lengths of the entries of $M$.} Then $\psm{A}{\alpha}\wh{x} = \lambda \wh{x}$ for some $\lambda \in \R$ with $\lambda < 0$. Now extend $\wh{x}$ by padding with zeros to $x \in \R^n$. Then for $i \in \alpha$, $x_i(Ax)_i = \lambda x_i^2 \leq 0$. For $i \notin \alpha$, $x_i = 0$ so $x_i(Ax)_i \leq 0$.
\end{proof}

\subsection{Sign-reversing vectors to non-positive minors}
The next few propositions follow immediately from \cite{gale1965jacobian}.

\begin{lemma} \label{lem1}
    If $A\in \R^{n\times n}$ is a P-matrix, then the only $x\in \R^n$ satisfying
    \begin{equation} Ax \leq 0, x \geq 0 \label{eq1}\end{equation} is the zero vector.\footcite[The lemma and proof follow Theorem 1 of][]{gale1965jacobian}
\end{lemma}
\begin{proof}
    We'll use induction on $n$. Assume that $x \in \R^n$ satisfies \eqref{eq1}, and that $A$ is a P-matrix. Then $A$ is invertible and $A^{-1}$ has at least one positive entry in each column. Let $b \in \R^n$ be the first column of $A^{-1}$. Let $\theta = \min_{i:b_i > 0} x_i/b_i$. This is well-defined since we are minimizing over a non-empty set. Let $k \in \argmin_{i:b_i > 0} x_i/b_i$ be the index of a component achieving the minimum. Define $y = x - \theta b$. By construction, $y_i = x_i - (x_k/b_k)b_i \geq 0$ for all $i$ and $y_k = 0$.
    Moreover, $Ay = Ax - \theta Ab$ so $(Ay)_i = (Ax)_i - \theta \One\Brack{i=1} \leq 0$ for all $i$. Let $\alpha = [n]\setminus \Set{k}$. Let $\wh{A} \coloneqq A_{\alpha,\alpha}$ and let $\wh{y} \coloneqq y_{\alpha}$. Then $\wh{A}$ is a P-matrix and $\wh{y}\geq 0$, $\wh{A}\wh{y} \leq 0$. By induction, $\wh{y} = 0$. Since $y_k = 0$, we have that $y = 0$ and $Ay = Ax - \theta Ab = 0$ so $Ax = \theta Ab$. For $i\neq 1$, $(Ax)_i = 0$; $(Ax)_1 = \theta$. But this is a contradiction, since $\theta > 0$ and $Ax \leq 0$.
\end{proof}

P-matrices are invariant under conjugation by signature matrices, where a \emph{signature matrix} is a diagonal matrix with all diagonal entries in $\Set{\pm 1}$. It will be useful in what follows to observe that signature matrices are orthogonal.

\begin{lemma} \label{lem:conjugation by sign matrix}
    Let $D$ be any signature matrix. Then $DAD$ is a P-matrix if and only if $A$ is a P-matrix. Furthermore, $Dx$ is a sign-reversing vector for $DAD$ if and only if $x$ is a sign-reversing vector for $A$.
\end{lemma}
\begin{proof}
    The proof of the latter of the two propositions implies the former. Fix a vector $x \in \R^n$ and a signature matrix $D$. Let $y = Dx$, and let $B = DAD$. Then
    \[ x_i(Ax)_i = x_i(D^2AD^2x)_i = (Dx)_i (DAD(Dx))_i = y_i(By)_i \] so $x_i(Ax)_i\leq 0 \iff y_i(By)_i \leq 0$.
\end{proof}

\begin{observation} \label{obs:strictly positive} Let $x \geq 0$ and define $\alpha \coloneqq \psupp(x)$. Then $x$ is a sign-reversing vector for $A$ if and only if $x_{\alpha}$ is a sign-reversing vector for $\psm{A}{\alpha}$.
\end{observation}

Using Observation~\ref{obs:strictly positive} and Lemma~\ref{lem:conjugation by sign matrix} we'll be able to reduce the task of transforming a witness to \ref{PV2} into a witness to \ref{PV1} to the simpler task of starting with a witness to
\begin{enumerate}[label=(PV\arabic*$^\ast$), start=2]
    \item A vector $x > 0$ such that $Ax \leq 0$. \label{PV2*}
\end{enumerate} and transforming it into a witness to \ref{PV1}.

\begin{corollary} Given a witness to \ref{PV2} and a procedure to obtain a witness to \ref{PV1} from \ref{PV2*} we can obtain a witness to \ref{PV1}.
\end{corollary}

\begin{proof} Let $x\in \R^n$ be a witness to \ref{PV2}, i.e., a sign-reversing vector. Let $D$ be the signature matrix where $D_{ii} = -1$ for $i \in \Setbar{i}{x_i < 0}$ and $D_{ii} = 1$ otherwise. Now $y \coloneqq Dx \geq 0$ is a sign-reversing vector for $B = DAD$. Letting $\alpha \coloneqq \psupp(y)$, we also have that $y_\alpha$ is a strictly-positive sign-reversing vector for $\psm{B}{\alpha}$. If we use the given procedure to obtain from this a non-positive minor of $\psm{B}{\alpha}$, let's say $\psm{B}{\beta}$ for $\beta \subseteq \alpha$, we can lift this back to a non-positive minor for $A$ by observing that $\psm{B}{\beta} = \psm{D}{\beta}\psm{A}{\beta}\psm{D}{\beta}$ so
    \[ \det(\psm{B}{\beta}) = \det(\psm{D}{\beta})\det(\psm{A}{\beta})\det(\psm{D}{\beta}) = \Paren{\det(\psm{D}{\beta})^2}\det(\psm{A}{\beta}) \] and $\det(\psm{B}{\beta}) \leq 0$ implies $\det(\psm{A}{\beta}) \leq 0$.
\end{proof}

We'll now abstract the operation used in the inductive proof of Lemma \ref{lem1}. The only requirement we had for the operation was that $M$ be invertible and that each column of $M^{-1}$ contains at least one positive entry.

\begin{definition}[Reducing a sign-reversing vector]
    Given a reducible matrix $A \in \R^{n\times n}$, a sign-reversing vector $x\in \R^n_{>0}$, and an index $i\in [n]$, we'll say that \emph{$x$ can be reduced at $i$ with respect to $M$} if $M$ is invertible and $\det(M_{\alpha,\alpha}) > 0$ where $\alpha \coloneqq [n]-i$. In this case we can define the following function: \[ \reduce_i(A,x) \coloneqq x - \Paren{\min_{j:A^{-1}_{ji}>0} x_i/A^{-1}_{ji}}A^{-1}_{\varindex i}.\]
    It will be convenient to define $\theta_i(A,x) \coloneqq \Paren{\min_{j:A^{-1}_{ji} > 0} x_i/A^{-1}}$ so that $\reduce_i(A,x) = x - \theta_i(A,x)A^{-1}_{\varindex i}$.
\end{definition}

\begin{lemma}
    Given a matrix $M\in \R^{n\times n}$ and a sign-reversing vector $x\in \R^n_{>0}$ such that $x$ can be reduced at $i$, the vector $y^i \coloneqq \reduce_i(M,x)$ is a non-negative sign-reversing vector for $M$. \label{lem:reduction preserves sign-reversing}
\end{lemma}
\begin{proof}
    Fix any $i\in [n]$. We first observe that $y^i = x - \theta_i(M,x)M^{-1}_{\varindex i} \geq 0$ by construction. Moreover,
    \begin{align*}
        My^i = Mx - M\theta_i(M,x)M^{-1}_{\varindex i} = Mx - \theta_i(M,x)e_i \leq 0
    \end{align*} since $\theta_i(M,x) > 0$ and $Mx \leq 0$.
\end{proof}

\begin{lemma}
    If $M\in \R^{n\times n}, n \geq 2$ is a reducible matrix and $x\in \R^n_{>0}$ is a strictly-positive sign-reversing vector, then for at least one $i\in [n]$, $y^i\coloneqq \reduce_i(M,x)$ is not the zero vector. \label{lem:y^i not all zero}
\end{lemma}
\begin{proof}
    Suppose $y^i = 0$ for all $i\in [n]$. Fix any distinct $i,j \in [n]$. Let $\theta^i \coloneqq \min \Set{x_i/M^{-1}_{ki}: M^{-1}_{ki} > 0}$, $\theta^j \coloneqq \min \Set{x_j/M^{-1}_{kj} : M^{-1}_{ki} > 0}$. Then \[x - \theta^i M^{-1}_{\varindex i} = y^i = 0 = y^j = x - \theta^j M^{-1}_{\varindex j}\] and $M^{-1}_{\varindex i} = (\theta^j / \theta^i) M^{-1}_{\varindex j}$. But then two columns of $M^{-1}$ are linearly dependent which contradicts the existence of $M^{-1}$.
\end{proof}

The significance of this operation is that it allows us to produce a sequence of strictly positive sign-reversing vectors with decreasing support:
\begin{lemma}
    Suppose $M\in \R^{n\times n}$ is a matrix, and let $x > 0$ be a sign-reversing vector for $M$. Then either there exists an $i\in [n]$ such that $y^i \coloneqq \reduce_i(M,x)$ is not the zero vector, or $\det(M) = 0$.
\end{lemma}
\begin{proof} We'll use induction on $n$. For $n=1$, $x > 0$ and $Mx < 0$ implies that $\det(M) = M_{11} < 0$. For $n > 1$ assume that $y^i = 0$ for all $i\in [n]$.
\end{proof}

\begin{theorem}
    Given a certificate of type \ref{PV2*}, there is a polynomial time procedure (shown in Algorithm~\ref{alg:find nonpositive minor}) to obtain a certificate of type \ref{PV1}.
\end{theorem}

\begin{algorithm}
    \caption{Algorithm for converting from \ref{PV2*} to \ref{PV1}} \label{alg:find nonpositive minor}
    \begin{algorithmic}
        \Require{$M \in \R^{n\times n}$, $\alpha \subseteq [n]$, a strictly-positive sign-reversing vector $x \in \R_{>0}^\alpha$ for $\psm{M}{\alpha}$}
        \Ensure{A subset $\beta\subseteq \alpha$ such that $\det(\psm{M}{\beta}) \leq 0$}
        \Procedure{RecursiveFindNonpositiveMinor}{$M,\alpha,x$}
        \State \textbf{if} $\Card{\alpha} = 1$ or $\det(\psm{M}{\alpha}) \leq 0$, \Return{$\alpha$}
        \For{$i\in \alpha$}
        \State \textbf{if} $\det(\psm{M}{\alpha - i}) \leq 0$, \Return{$\alpha - i$}
        \State $y^i \gets \reduce_i(\psm{M}{\alpha}, x)$.
        \EndFor
        \State $i_* \gets \min\Set{i\in \alpha: y^i \neq 0}$.
        \State $y \gets y^{i_*}$.
        \State $\beta \gets \psupp(y)$.
        \State \Return{\textsc{RecursiveFindNonpositiveMinor}$(M,\beta, y_{\beta})$}
        \EndProcedure

        \Procedure{FindNonpositiveMinor}{$M,x$}
        \State \Return{\textsc{RecursiveFindNonpositiveMinor}$(M,[n],x)$}
        \EndProcedure
    \end{algorithmic}
\end{algorithm}

\begin{proof}
    First, the algorithm computes at most $(1+n) + (1 + (n-1)) + \dotsb + 1 = O(n^2)$ determinants across the at-most $n$ recursive calls, each of which can be done in $\poly(n)$ time.

    To establish correctness, we observe that if $x$ is a strictly positive sign-reversing vector for $M_{\alpha,\alpha}$, then $i_*$ is guaranteed to exist by Lemma~\ref{lem:y^i not all zero}, and by Lemma~\ref{lem:reduction preserves sign-reversing}, $y$ is a non-negative sign-reversing vector for $\psm{M}{\alpha}$. By construction, $y_{\beta}$ is strictly positive and sign-reversing on $\psm{M}{\beta}$, proving that the recursive call preserves the invariants. Finally, we observe that we only return from the algorithm if we've directly witnessed a subset $\beta$ with $\det(\psm{M}{\beta}) \leq 0$ or if $\Card{\alpha} = 1$. In the latter case, $\det(\psm{M}{\alpha}) = M_{ii}$ and $x$ being a sign-reversing vector means that $M_{ii}x_i \leq 0$ while $x_i > 0$, so $M_{ii} \leq 0$ and $\det(\psm{M}{\alpha}) \leq 0$.
\end{proof}

\section{Converting from a two-distinct sinks witness} Fix some non-degenerate $M$, and a vector $q \in \R^{d}$. Now for any $\alpha \subseteq d$, there is a unique pair $(w;z)$ such that $Iw - Mz = q$ and $\supp(w) \subseteq \alpha$ and $\supp(z) \subseteq [n]\setminus \alpha$. Note that we may not have $w,z\geq 0$ for a given $\alpha$.

Now assume we are given subsets $\alpha, \beta \subseteq [n]$ such that
\[ [\Char(\alpha) \xor \Char(\beta)] \cap [\outmap(\alpha) \xor \outmap(\beta)] = 0^n\text{.} \]

Let $(w^{(\alpha)}; z^{(\alpha)})$ and $(w^{(\beta)}; z^{(\beta)})$ be the solutions induced by $\alpha$ and $\beta$ respectively.
\section{Non-positive minors and multiples sinks}\label{sec:minor2uso}
\subsection{Non-positive minors implies Unique Sink Orientation Violation}
We begin by obtaining from our non-positive minor another non-positive minor containing a strictly positive minor of cardinality one smaller:
\begin{lemma}\label{lem:minimalWitness} Given a violation of type \ref{PV1}, i.e. a subset $\alpha$ such that $\det(\psm{A}{\alpha}) \leq 0$, we can find a subset $\alpha' \subseteq \alpha$ so that either \begin{itemize}
        \item $\det(\psm{A}{\alpha'}) = 0$, or
        \item $\det(\psm{A}{\alpha'}) < 0$ and for some $i \in \alpha'$, $\det(\psm{A}{\alpha' - i}) > 0$ in time $O(n^5)$.
    \end{itemize}
\end{lemma}

\begin{proof} If at any point we compute a determinant for some subset $\alpha'$ and obtain $\det(\psm{A}{\alpha'}) = 0$, we return $\alpha'$.
    In the event that $\alpha$ has exactly one element, then we can return $\alpha$, as $\det(\psm{A}{\alpha}) \le 0$ by assumption, and the only proper subset of $\alpha$ is $\varnothing$, where $\det(\psm{A}{\varnothing}) = 1$.

    Otherwise, for each $i \in \alpha$, compute $\det(\psm{A}{\alpha - i})$, and if it is positive, return $\alpha$ and $i$. Otherwise, recursively apply this algorithm for $\alpha = \alpha - i$, for any $i \in \alpha$.

    Each determinant computation requires $O(n^3)$ time, and we compute at most $n^2$ determinants.
\end{proof}

We can now proceed to the main reduction.

\begin{theorem}
    We can convert a violation of type \ref{PV1} to a violation of type \ref{PV3} in polynomial time.
\end{theorem}
\begin{proof}

    Given a subset $\alpha$ with $\det(\psm{M}{\alpha'}) \leq 0$ as in \ref{PV1}, we wish to either find an $\alpha$ so that $\amat{\alpha}{M}$ is not invertible, or find a vector $q\in \R^n$ and two distinct sets $\alpha, \beta \subseteq [n]$ such that
    \[\Brack{\Char(\alpha) \xor \Char(\beta)} \cap \Brack{\outmap_q(\alpha, M) \xor \outmap_q(\beta, M)} = 0^n\text{.} \]

    We begin by applying the procedure in Lemma~\ref{lem:minimalWitness}. If we obtain an $\alpha'$ such that $\det(\psm{A}{\alpha'}) = 0$, we can return $\alpha'$ since $\det(\amat{\alpha'}{M}) = (-1)^{\Card{\alpha'}}\det(\psm{M}{\alpha'}) = 0$. If instead we obtain an $\alpha'$ and $i$ such that $\det(\psm{M}{\alpha'}) < 0$ and $\psm{M}{\alpha' - i} > 0$, we continue as follows.

    We now claim that if we take $\alpha'$ and $i$ as in the above lemma, and let $\beta \coloneqq \alpha' - i$, we will in fact obtain that $\alpha'$ and $\beta$ satisfy our desired property for any nondegenerate $q \in \R^n$. To simplify notation slightly, we'll rename $\alpha'$ to $\alpha$ and forget the original $\alpha$ we were given.

    To see this, it will be helpful to first divide $\amat{\alpha}{M}$ and $\amat{\beta}{M}$ into blocks. By reordering the basis elements in our representation, we will take $\alpha$ to be the first $|\alpha|$ rows.
    \[
        \amat{\alpha}{M} =
        \begin{pmatrix}
            -M|_{\beta}            & -M_{\beta i}       & 0                 \\
            -M_{\bar{\beta} \beta} & -M_{\bar{\beta} i} & I_{|\bar{\beta}|}
        \end{pmatrix}
    \]
    \[
        \amat{\beta}{M} =
        \begin{pmatrix}
            -M|_{\beta}            & 0                 & 0                 \\
            -M_{\bar{\beta} \beta} & e_{|\bar{\beta}|} & I_{|\bar{\beta}|}
        \end{pmatrix}
    \]

    As in the proof of Lemma \ref{lem:nondegenerate}, we can compute the determinants as
    $\det(\amat{\alpha}{M}) = (-1)^{|\alpha|}\det(\psm{M}{\alpha})$, and
    $\det(\amat{\beta}{M})= (-1)^{|\beta|}\det(\psm{M}{\beta})$.

    Now, by Cramer's rule, we can compute

    \[
        (\amat{\alpha}{M}^{-1}q)_i =
        \frac{1}{\det(\amat{\alpha}{M})}
        \begin{pmatrix}
            -M|_{\beta}            & q_{\beta}         & 0                 \\
            -M_{\bar{\beta} \beta} & q_{\bar{\beta} i} & I_{|\bar{\beta}|}
        \end{pmatrix}
    \]
    \[
        (\amat{\beta}{M}^{-1}q)_i =
        \frac{1}{\det(\amat{\beta}{M})}
        \begin{pmatrix}
            -M|_{\beta}            & q_{\beta}         & 0                 \\
            -M_{\bar{\beta} \beta} & q_{\bar{\beta} i} & I_{|\bar{\beta}|}
        \end{pmatrix}
    \]

    In particular, assuming that $q$ is nondegenerate,
    \begin{align*}
        \frac{(\amat{\alpha}{M}^{-1}q)_i}{(\amat{\beta}{M}^{-1}q)_i}
         & =\frac{\det(\amat{\beta}{M})}{\det(\amat{\beta}{M})}                         \\
         & =(-1)^{|\alpha| - |\beta|}\frac{\det(\psm{M}{\alpha})}{\det(\psm{M}{\beta})} \\
         & = \frac{-\det(\psm{M}{\alpha})}{\det(\psm{M}{\beta})}                        \\
         & >0
    \end{align*}
    where this last inequality is based on our choice of $\alpha$ and $\beta$.

    In particular $\out_q(\alpha, M)_i = \out_q(\beta, M)_i$. Thus, we obtain
    \[
        \Char(\alpha) \oplus \Char(\beta) = e_i; \quad
        (\out_q(\alpha, M) \oplus \out_q(\beta, M))_i = 0.
    \]
    Therefore, $(\Char(\alpha) \oplus \Char(\beta)) \cap (\out_q(\alpha, M) \oplus \out_q(\beta, M))_i = 0^n$.

    Finally, assuming that $M$ is nondegenerate, and the entries of $M$ are rational numbers, we can find a nondegenerate $q$ in polynomial time. To do this, we rely on a lemma from \cite{schrijver1998theory}.
    \begin{lemma}
        Suppose that $X$ is an $n\times n $ matrix of rational numbers, and that $X$ can be expressed in binary using at most $\sigma$ bits (by expressing the numerator and denominator of each entry of $X$ in binary). Then $\det(X)$ is a rational number that can be expressed in binary using at most $2\sigma$ binary bits.
    \end{lemma}
    In particular, if $X$ can be expressed using at most $\sigma$ bits, then $\det(X)$ is a rational number whose denominator is at most $\frac{1}{2^{2\sigma}}$.

    We see therefore that for any (possibly nonprincipal) submatrix of $M$ whose size is bounded by $\sigma$, $M_{\alpha \beta}$, either $|\det(M_{\alpha \beta})| \ge \frac{1}{2^{2\sigma}}$ or $\det(M_{\alpha \beta})=0$. Moreover, $|\det(M_{\alpha \beta})| < 2^{2\sigma}$.

    Consider therefore the vector $q$ where $q_i = 2^{-4i\sigma}$. This can clearly be computed in polynomial time, since each entry can easily be expressed in binary using at most $4d\sigma$ bits. We claim that as long as $M$ is nondegenerate, $q$ is in fact nondegenerate for $M$. This is true essentially because the entries of $q$ decrease rapidly enough that as long as $\det(M|_{\alpha}) \neq 0$ for all $\alpha$, it is not possible for $(\amat{\alpha}{M}^{-1}q)_i = 0$.

    Formally, we use Cramer's rule to state that
    \[
        (\amat{\alpha}{M}^{-1}q)_i = \frac{(-1)^{|\alpha|}}{\det(M_{\alpha})}\sum_{j=1}^d \det(M_{\alpha-i \alpha-j})q_j.
    \]
    Suppose that this is 0.

    From our assumption that $\det(M_{\alpha-i \alpha-i}) \neq 0$, we know that for at least 1 $j$, $\det(M_{\alpha-i, \alpha-j}) \neq 0$, so we take $j^*$ to be the smallest $j$ so that $\det(M_{\alpha-i, \alpha-j}) \neq 0$, and then notice that
    \begin{align*}
        \Abs{\sum_{j=1}^d \det(M_{\alpha-i, \alpha-j})q_j}
         & =\Abs{\sum_{j=j^*}^d \det(M_{\alpha-i, \alpha-j})q_j}                                                 \\
         & \ge\Abs{\det(M_{\alpha-i, \alpha-j^*})q_{j^*}}-\Abs{\sum_{j=j^*+1}^d \det(M_{\alpha-i, \alpha-j})q_j} \\
         & \ge\Abs{2^{-4(j^*+\frac{1}{2})\sigma}}-\Abs{\sum_{j=j^*+1}^d 2^{-2(j-\frac{1}{2})\sigma}}             \\
         & >0
    \end{align*}
\end{proof}

\subsection{Unique Sink Orientation Violation implies Negative Minor}
\begin{theorem}
    We can convert a violation of type \ref{PV3} to a violation of type \ref{PV1} in polynomial time.
\end{theorem}
\begin{proof}
    Suppose we are given a subset $\alpha$ such that $\det(\amat{\alpha}{M}) = 0$. In that case, we can return $\alpha$ immediately.

    Suppose instead we are given as input a pair of subsets $\alpha, \beta \subseteq [d]$ and a vector $q\in \R^d$  with the property that $(\Char(\alpha) \oplus \Char(\beta)) \cap (\out_q(\alpha, M) \oplus \out_q(\beta, M)) = 0^d$.

    Let $S = \alpha \cap \beta$, and let $T = \alpha \cup \beta$. We want to reduce to the case when $S = \varnothing$ and $T = [n]$. One way to think about this in terms of unique sink orientations is that we are restricting our attention to the minimal face of the hypercube containing both $\alpha$ and $\beta$.

    First,  we can consider the submatrix $M' = M|_{T}$, and the subvector $q' = q_T$. It is clear that if we can find a nonpositive minor of $M'$, that directly translates to a nonpositive minor of $M$, and moreover, it can easily be seen from our assumptions that $\Char(\alpha) \oplus \Char(\beta)) \cap (\out_{q'}(\alpha, M') \oplus \out_{q'}(\beta, M')) = 0^d$. So, we may as well assume that $T = [d]$ from now on.

    To reduce to $S = \alpha \cap \beta$, we will need to do a little more work.

    To do this, we define a linear map $L$ as follows. For $i \in S$, we let
    \[
        L M_{\varindex i} = e_i,
    \]
    and for $i \not \in S$,
    \[
        L e_i = e_i
    \]
    Notice that such an $L$ exists if and only if $\det(M|_{S}) \neq 0$, and that we can check for the existence of this $L$ by solving a system of linear equations. Moreover, if such an $L$ exists, then its determinant is precisely $\frac{1}{\det(M|_S)}$, which we can compute, and assume is positive for the remainder.

    If we take $S$ to be $\{1 ,\dots, k\}$, then
    \[
        LM =
        \begin{pmatrix}
            I_k & G        \\
            0   & M|_{S^c}
        \end{pmatrix},
    \]
    for some $k \times (d-k)$ matrix $G$.

    Now, we let $M' = (LM)|_{S^c} = M|_{S^c}$, and $q' = (Lq)_{S^c}$. Let $\alpha' = \alpha \setminus S$, and $\beta' = \beta \setminus S$.
    Once again, based on the above decomposition, we see that a nonpositive minor of $M'$ will directly translate to a nonpositive minor of $M$. Now, we should show that
    \[
        (\Char(\alpha') \oplus \Char(\beta')) \cap (\out_{q'}(\alpha', M') \oplus \out_{q'}(\beta', M')) = 0^d
    \]
    Hence, assume that $i \in \alpha' \setminus \beta'$.

    We will denote by $M / q,i$ the matrix obtained by replacing the $i^{th}$ column of $M$ by $q$, so that Cramer's rule states that $(M^{-1}q))_i = \frac{1}{\det(M)} \det(M/ q, i )$. We apply Cramer's rule:
    \[
        (\amat{\alpha'}{M'}^{-1} q')_i = \frac{1}{\det(\amat{\alpha'}{M'})} \det(\amat{\alpha'}{M'}/ q', i )
    \]
    We can expand this out by blocks to see that $ \det(\amat{\alpha'}{M'}/ q', i ) =  (-1)^{|S|}\det(L) \det(\amat{\alpha}{M}/ q, i )$, and analogously, we have that $ \det(\amat{\beta'}{M'}/ q', i ) =  (-1)^{|S|}\det(L) \det(\amat{\beta}{M}/ q, i )$ 

    Now, we wish to reduce to the case that $\out_q(\alpha, M) = \out_q(\beta, M) = 1^n$.

    To do this, we will once again use the notion of a signature matrix, defined in Section~\ref{sec:minor2vector}. Let $D$ be the signature matrix where $D_{ii} = -1$ if $\out_q(\alpha, M)_i = 0$ and $D_{ii} = 1$ otherwise.

    It is not hard to see that
    \[
        \amat{\alpha}{DMD} = D\amat{\alpha}{M}D,
    \]
    \[
        \amat{\beta}{DMD} = D\amat{\beta}{M}D.
    \]

    So, we may consider
    \[
        \amat{\alpha}{DMD}^{-1}Dq = D\amat{\alpha}{M}^{-1}D(Dq) = D(\amat{\alpha}{M}^{-1}q),
    \]
    \[
        \amat{\beta}{DMD}^{-1}Dq = D\amat{\beta}{M}^{-1}D(Dq) = D(\amat{\beta}{M}^{-1}q).
    \]

    We see then that if we replace $M$ with $DMD$, and $q$ by $Dq$, we can assume $\out_q(\alpha) = \out_q(\beta) = 1^n$.

    Now, we can conclude by using the next lemma.

\end{proof}
\begin{lemma}
    There is a polynomial time algorithm, which, given an input matrix $M$, a vector $q$ and sinks $\alpha$, $\beta \subseteq [n]$ with $\alpha \neq \beta$, will output a subset $\gamma$ so that the submatrix $\det(\psm{M}{\gamma}) \le 0$.
\end{lemma}
\begin{proof}
    We will consider two cases: if $\alpha$ and $\beta$ differ by exactly one element, say $\alpha = \beta + i$, we have already seen that $\frac{(\amat{\alpha}{M}^{-1}q)_i}{(\amat{\beta}{M}^{-1}q)_i} = -\frac{\det(\psm{M}{\alpha})}{\det(\psm{M}{\beta})}$ by Cramer's rule.
    Therefore, $\frac{\det(\psm{M}{\alpha})}{(\det(\psm{M}{\beta})}$ is negative and one of $\psm{M}{\alpha}$ and $\psm{M}{\beta}$ has negative determinant. A simple determinant computation will give us $\gamma$ is either $\alpha$ or $\beta$.

    On the other hand, if $\alpha$ and $\beta$ have more elements in their symmetric difference, then we can find the desired $\gamma$ recursively using the following lemma:
    \begin{lemma}\label{lem:newsink}
        Suppose that $\alpha$ and $\beta$ are sinks with respect to $q$ and $\alpha \setminus \beta$ is not empty. Further suppose that $\det(\psm{M}{\alpha})> 0$, $\det(\psm{A}{\beta}) > 0$ and for all $i \in \alpha \setminus \beta$, $\det(\psm{M}{\beta + i}) > 0$.
        Then there is some $i \in \alpha \setminus \beta$ so that $\beta + i$ is a sink with respect to some vector $q'$ where $q'$ can be found using $O(1)$ linear algebra operations.
    \end{lemma}
    Using this lemma, we can repeatedly find $i \in \alpha \setminus \beta$ and let $\beta = \beta + i$ and repeat the algorithm.
    This step only fails if we find that either $\alpha$, $\beta$ or $\beta + i$ corresponds to a submatrix of $A$ with nonpositive determinant, in which case we can terminate and report the appropriate submatrix.

    We continue this process until $\alpha \subseteq \beta$, and then can replace $\alpha$ by $\beta$ and continue until we reach the case when $\alpha = \beta + i$, and then we are in our base case.

    It is clear that we require computing at most $O(n)$ basic linear algebra operations in the process of performing this algorithm, showing that the algorithm is polynomial time.
\end{proof}

The proof that follows is essentially the same as in \cite{samelson1958partition}, except using matrix notation and expanded upon slightly.
\begin{proof}{Proof of Lemma \ref{lem:newsink}}

    The idea is the following: we will find $v$ with 3 properties:
    \begin{enumerate}
        \item $\amat{\alpha}{M}^{-1} v$ is nonnegative.
        \item For each $i \in (\beta \Delta \alpha)^c$, $\amat{\beta}{M}^{-1} v > 0$, where $\alpha \Delta \beta$ is the symmetric difference of $\beta$ and $\alpha$, i.e., $\alpha \Delta \beta \coloneqq (\alpha \setminus \beta)\cup (\beta \setminus \alpha)$.
        \item There is some $i \in \alpha \setminus \beta$ so that $(\amat{\beta}{M}^{-1} v)_i < 0$.
    \end{enumerate}
    Given this, we first consider $q' = q + tv$, where $t>0$ is chosen such that $\amat{\beta}{M}^{-1}q'$ is nonnegative, but for some $i \in \alpha \setminus \beta$, $(\amat{\beta}{M}^{-1}q')_i = 0$.

    Given this fact, consider $\beta + i$, and notice that $\amat{\beta + i}{M}$ differs from $\amat{\beta}{M}$ in exactly column $i$.
    From this, it can be seen that for any vector $w$ so that $w_i = 0$, $\amat{\beta+i}{M}w = \amat{\beta}{M}$.
    In particular, $\amat{\beta+i}{M}^{-1}q' = \amat{\beta}{M}^{-1}q'$.

    Now, because of property 1 of $v$, $\amat{\alpha}{M}^{-1}q'$ is positive, so by perturbing $q'$ slightly (say adding a small multiple of the sum of the columns of $\amat{\beta+i}{M}$), we can make $\amat{\beta+i}{M}^{-1}q'$ positive, while keeping $\amat{\alpha}{M}^{-1}q'$ positive, which will imply that $\beta +i$ and $\alpha$ are both sinks for this value of $q'$.

    To find $v$, we will do the following: firstly, we consider $M_{\varindex, i}$, the $i^{th}$ column of $A$.
    It is clear that $-M_{\varindex, i}$ satisfies property 1, since $\amat{\alpha}{M}(-M_{\varindex, i}) = e_i$.

    To see that $-M_{\varindex, i}$ also satisfies property 3, we use Cramer's rule again, to see that
    \[
        (\psm{M}{\beta}^{-1}M_{\varindex, i})_i = \frac{\det(\psm{M}{\beta + i})}{\det(\psm{M}{\beta})}.
    \]
    By our assumptions, this is positive, and so property 3 holds.

    To get our desired $v$, we need to alter $M_{\varindex, i}$ to satisfy property 2. To do this, let $a = \sum_{j \in \alpha \Delta \beta} M_{\varindex, j}$.
    Notice that $\amat{\beta}{M}^{-1}a = \sum_{j \in \alpha \Delta \beta} e_j$.
    In particular, if we add a large enough multiple of $a$ to $-M_{\varindex, i}$, we will obtain a vector that satisfies all 3 properties.

    This concludes the proof.
\end{proof}

\printbibliography

\end{document}